\documentclass{article}

\usepackage{amsmath,amsfonts,amsthm}

\title{Identifying the simple finite-dimensional Lie algebras over
  $\mathbb C$ by means of simple sequences}
\author{Kai Neerg\aa rd}
\date{}

\theoremstyle{plain}\newtheorem{lemma}{Lemma}
\theoremstyle{plain}\newtheorem{proposition}{Proposition}
\theoremstyle{plain}\newtheorem{theorem}{Theorem}
\theoremstyle{remark}\newtheorem*{remark}{Remark}

\begin{document}

\maketitle

\begin{abstract}
  A novel method of determining which Dynkin diagrams represent simple
  finite-dimensional Lie algebras over $\mathbb C$ is presented. It is
  based on a condition that is both necessary and sufficient for a
  suitably defined Cartan matrix to be expressible by scalar products
  in a Euclidean vector space. The sufficiency of this condition makes
  unnecessary subsequent verification of the existence of a Lie
  algebra or root system corresponding to each Dynkin diagram by
  explicit construction. The Dynkin diagrams are selected by
  examination of an easily calculated sequence of minors of a
  symmetrised Cartan matrix. These minors are mostly integers.
\end{abstract}

\section{\label{sec:in}Introduction}

The simple finite-dimensional Lie algebras over $\mathbb C$ were
classified in the last decades of the 19th century by
Killing~\cite{ref:Kil88} and Cartan~\cite{ref:Car94}. In their
analysis \textit{roots} play a central role. These are vectors in a
Euclidean vector space of dimension $l$, where $l$ is the rank of the
Lie algebra. For their definition I refer to the literature,
e.g.~\cite{ref:Jac62,ref:Ser66,ref:Hum72}. The Bourbaki group set up
axioms for a \textit{root system}, which is a finite set of vectors in
a finite-dimensional Euclidean vector space, and defined a subclass of
\textit{reduced root systems}. A root system is \textit{irreducible}
if it cannot be split into mutually orthogonal subsystems. A reduced
root system has the properties of the set of roots of a semisimple
finite-dimensional Lie algebra over $\mathbb C$, and it was shown by
Serre that there is a 1--1 correspondence up to isomorphisms between
reduced root systems and such Lie algebras~\cite{ref:Ser66}. An
equivalent result can be derived from theorem 7.5 in~\cite{ref:Jac62}.
The irreducible reduced root systems correspond to simple Lie
algebras. A \textit{simple system of roots} is a subset of a reduced
root system which is a basis for the ambient Euclidean vector space
such that every member of the root system is a non-zero linear
combination of the basic vectors with integer coefficients that are
either all non-negative or all non-positive. The simple system of
roots of a given reduced root system is unique up to transformations
in the symmetry group of the root system.

Upon an ordering $\alpha_i,i= 1 \ldots l,$ of a simple system of roots
it determines an $l \times l$ \textit{Cartan matrix} $A$ with entries
\begin{equation}\label{eq:A}
  A_{ij} = \frac{2 (\alpha_i,\alpha_j)}{(\alpha_i,\alpha_i)} ,
\end{equation}
where $(,\!)$ denotes the scalar product of the ambient Euclidean
vector space. The Cartan matrix determines the reduced root system up
to isomorphisms~\cite{ref:Ser66}. It has the following properties. (i)
The entries $A_{ij}$ are integers. (ii) The diagonal entries $A_{ii}$
equal 2. (iii) The off-diagonal entries, $A_{ij}, i \ne j,$ are
non-positive. (iv) For $i \ne j$ the product $A_{ij} A_{ji}$ is an
integer in the range from 0 to 3 and $A_{ji} = 0$ if $A_{ij} = 0$. I
call \emph{every} matrix with these properties a Cartan matrix. A
Cartan matrix then is not necessarily given by \eqref{eq:A} in terms
of vectors $\alpha_i$ in a Euclidean vector space. Cartan matrices are
considered \textit{isomorphic} if they arise from one another by a
transformation $A \mapsto S A S^T$, where $S$ is a permutation matrix.
If they are Cartan matrices of a reduced root system this corresponds
to reordering of any chosen simple system of roots. The standard
analysis of the structure of semisimple finite-dimensional Lie
algebras over $\mathbb C$ or reduced root
systems~\cite{ref:Jac62,ref:Ser66,ref:Hum72} involves ruling out some
Cartan matrices as incompatible with~\eqref{eq:A} with $(,)$ a
Euclidean scalar product by demanding that specific linear
combinations of the basic vectors $\alpha_i$ obey inequalities valid
in a Euclidean vector space. This results in a list of ``admissible''
isomorphism classes of Cartan matrices. The existence of a Lie algebra
or reduced root system corresponding to each admissible isomorphism
class is then verified subsequently by explicit construction.

I take in this note a different path by applying a criterion that is
both necessary and \emph{sufficient} for a Cartan matrix to be given
by~\eqref{eq:A} with the scalar products $(\alpha_i,\alpha_j)$ derived
from a positive definite quadratic form $(\alpha,\alpha)$ on a vector
space over $\mathbb R$ with basic vectors $\alpha_i,i= 1 \ldots l$.
This makes unnecessary subsequent verification of the existence of a
corresponding Lie algebra or reduced root system by explicit
construction. Indeed, given an ordered basis
$(\alpha_i,i= 1 \ldots l)$ for a vector space over $\mathbb R$ any
Cartan matrix $A$ generates a, possibly infinite, set $R$ of linear
combinations $\alpha = \sum_i k_i \alpha_i$ with integer coefficients
$k_i$ by arbitrary repeated application of the Weyl reflections
$w_i:\alpha \mapsto \alpha - \sum_j k_j A_{ij} \alpha_i$ starting from
the basis. The transformations $w_i$ map $R$ into itself. When $A$ is
given by \eqref{eq:A} in terms of a positive definite quadratic form
$(\alpha,\alpha)$ the transformation $w_i$ becomes the reflection in
the hyperplane through 0 perpendicular to $\alpha_i$ with respect to
this quadratic form. Since these transformations are orthogonal every
vector $\alpha \in R$ obeys
$(\alpha,\alpha) \le \max \{(\alpha_i,\alpha_i) |i= 1 \ldots l\}$. The
set $R$ is therefore finite. Since the transformations $w_i$ are
involutions they generate a group $W$ of orthogonal transformations. A
Weyl reflection $w_\alpha$ can be defined for every $\alpha \in R$ as
the reflection in the hyperplane through 0 perpendicular to $\alpha$,
and because every $w_\alpha$ is given by $w_\alpha = s w_i s^{-1}$ for
some $i,1 \le i \le l,$ and some $s \in W$ the set $R$ is invariant
under every $w_\alpha$. Thus $R$ is a root system. If
$\alpha,\beta \in R$ obey $\beta = t \alpha$ with $|t| < 1$ and
$i,1 \le i \le l,$ and $s \in W$ are such that $\alpha = s \alpha_i$
then $s^{-1} \beta = t\alpha_i$, contradicting that every
$\gamma \in R$ is an integer linear combination of the basic vectors
$\alpha_i$. Thus $R$ is reduced. It is another, in my view, attractive
feature of the method to be presented that it does not rely on ad hoc
constructions in individual cases but applies one universal criterion
to every potential Cartan matrix of a simple finite-dimensional Lie
algebra over $\mathbb C$.

My method is based on the observation that a Cartan matrix being given
by \eqref{eq:A} in terms of an ordered basis
$(\alpha_i,i= 1 \ldots l)$ for a vector space $V$ over $\mathbb R$ and
a positive definite quadratic form $(\alpha,\alpha)$ on $V$ is
equivalent to positive definiteness of a \textit{symmetrised Cartan
  matrix} to be defined in section~\ref{sec:sym}. The method makes
use, moreover, of the fact that in many relevant cases the Cartan
matrix and its symmetrised version can be chosen tridiagonal. In the
remaining cases a suitably chosen symmetrised Cartan matrix is
rendered so by a fairly simple orthogonal transformation. Applying
Sylvester's criterion for positive definiteness of a real symmetric
matrix then amounts to examining a simple sequence of mostly integers
which is calculated quickly from the matrix in question.

\section{\label{sec:sym}Symmetrised Cartan matrices. Dynkin and
  Coxeter diagrams}

This section is devoted to some preliminary observations and
definitions before the main discussion to follow in
section~\ref{sec:an}. Consider a Cartan matrix $A$ and let it be
assumed that every pair of different indices $i_0$ and $i_n$ is
connected by a sequence $i_0,i_1,\dots,i_n$ such that
$A_{i_0i_1} A_{i_1i_2} \cdots A_{i_{n-1}i_n} \ne 0$. Then also
$A_{i_1i_0} A_{i_2i_1} \cdots A_{i_ni_{n-1}} \ne 0$. For $A$ to be
given by~\eqref{eq:A} in terms of an ordered basis
$(\alpha_i,i= 1 \ldots l)$ for a vector space $V$ over $\mathbb R$ and
a quadratic form $(\alpha,\alpha)$ on $V$ none of the squared norms
$(\alpha_i,\alpha_i)$ must equal zero. Since $A$ does not depend on
the overall normalisation of $(\alpha,\alpha)$ one can therefore
choose any one of them, say $(\alpha_{i_0},\alpha_{i_0})$, to be
positive. For every $i_n \ne i_0$ one then has
$(\alpha_{i_n},\alpha_{i_n}) = (A_{i_0i_1} \cdots A_{i_{n-1}i_n}) /
(A_{i_1i_0} \cdots A_{i_n i_{n-1}}) (\alpha_{i_0},\alpha_{i_0}) > 0$
in terms of a sequence $i_0,i_1,\dots,i_n$ as above. Every
$(\alpha_i,\alpha_i)$ thus being determined up to the overall factor
$(\alpha_{i_0},\alpha_{i_0})$, the scalar products
$(\alpha_i,\alpha_j)$ with $i \ne j$ are given by~\eqref{eq:A}. In
this way $A$ determines every scalar product $(\alpha_i,\alpha_j)$ up
to an overall non-zero factor, and this factor may be chosen such that
every $(\alpha_i,\alpha_i)$ is positive. If the set of indices is the
disjoint union of subsets internally connected by sequences
$i_0,i_1,\dots,i_n$ as above but not mutually connected by such
sequences then an arbitrary non-zero factor applies to the scalar
products $(\alpha_i,\alpha_j)$ within each subset and
$(\alpha_i,\alpha_j) = 0$ when $i$ and $j$ belong to different
subsets. Because the vectors $\alpha_i$ span $V$ the matrix $A$ thus
determines the entire quadratic form $(\alpha,\alpha)$ up to an
arbitrary overall non-zero factor applied to the scalar products
$(\alpha,\beta)$ within each subspace spanned by the vectors
$\alpha_i$ with indices $i$ in one of the subsets, and these factors
may be chosen such that every $(\alpha_i,\alpha_i)$ is positive. When
$\alpha$ and $\beta$ belong to different subspaces
$(\alpha,\beta) = 0$. When $A$ is a Cartan matrix of a reduced root
system each subspace corresponds to a an irreducible component of the
root system, so when the root system is irreducible all the indices
$i$ are connected by sequences $i_0,i_1,\dots,i_n$ as above and
$(\alpha,\alpha)$ is determined by the Cartan matrix up to one overall
non-zero factor.

When $A$ is given by~\eqref{eq:A} and all the scalar products
$(\alpha_i,\alpha_i)$ are positive the real $l \times l$ matrix $B$
with entries
\begin{equation}\label{eq:B1}
  B_{ij} = (\alpha_i,\alpha_i)^{\frac12} A_{ij}
      (\alpha_j,\alpha_j)^{-\frac12}
    = \frac{2 (\alpha_i,\alpha_j)}
      {\sqrt{(\alpha_i,\alpha_i) (\alpha_j,\alpha_j)}} .
\end{equation}
is symmetric. Conversely if positive real numbers $c_i$ exist such
that the $l \times l$ matrix $B$ with entries
\begin{equation}
  B_{ij} = c_i A_{ij} c_j^{-1}
\end{equation}
is symmetric then $(A_{i_0i_1} \cdots A_{i_{n-1}i_n}) /(A_{i_1i_0}
\cdots A_{i_ni_{n-1}}) = c_{i_n}^2/c_{i_0}^2$ for the \linebreak
sequences $i_0,i_1,\dots,i_n$ considered above, so the determination
of \linebreak $(\alpha_{i_n},\alpha_{i_n})
/(\alpha_{i_0},\alpha_{i_0})$ is unambiguous. Then $A$ is given by
\eqref{eq:A} in terms of an ordered basis $(\alpha_i,i= 1 \ldots l)$
for a vector space $V$ over $\mathbb R$ with the scalar products
$(\alpha_i,\alpha_j)$ derived from a quadratic form $(\alpha,\alpha)$
on $V$ with every $(\alpha_i,\alpha_i)$ positive. The existence of
such $c_i$ is thus a necessary and sufficient condition for a Cartan
matrix $A$ to be expressible in this way. It is then also a necessary
condition for the Cartan matrix being derivable from a reduced root
system. When such positive real numbers $c_i$ exist I call $A$
\textit{symmetrisable} and I call $B$ the corresponding
\textit{symmetrised Cartan matrix}. Symmetrised Cartan matrices are
implicit in the discussions in~\cite{ref:Jac62,ref:Ser66,ref:Hum72}.
Isomorphism of symmetrised Cartan matrices is defined as for Cartan
matrices. Clearly either all or none of the members of an isomorphism
class are positive definite. It follows from the last expression
in~\eqref{eq:B1} that if $\alpha = \sum_{i=1}^l x_i \alpha_i$ with
real coefficients $x_i$ then $(\alpha,\alpha) = \frac12 v B v^T$,
where $v$ is the row vector with entries $x_i
\sqrt{(\alpha_i,\alpha_i)}$. Therefore $(\alpha,\alpha)$ is positive
definite if and only if $B$ is so.

The \textit{Dynkin diagram}~\cite{ref:Dyn47} of a Cartan matrix $A$ is
a graph with a vertex for each row in $A$ and $m_{ij} = A_{ij} A_{ji}$
lines between the vertices corresponding to the $i$th and $j$th rows,
where $i \ne j$. The collection of $m_{ij}$ lines between the two
vertices is considered a single line with \textit{multiplicity}
$m_{ij}$. The definition of a Cartan matrix implies that the
multiplicity is an integer in the range from 1 to 3. Lines with
$m_{ij} > 1$ are called \textit{multiple} lines. A direction is
assigned to each multiple line to show which of $A_{ij}$ and $A_{ji}$
is the larger. Clearly Cartan matrices have identical Dynkin diagrams
if and only if they are isomorphic, so each Dynkin diagram describes
an isomorphism class of Cartan matrices. An undirected Dynkin diagram
is called a \textit{Coxeter diagram}~\cite{ref:Ser66} and describes an
isomorphism class of symmetrised Cartan matrices. The lines of the
Coxeter diagram of a symmetrised Cartan matrix $B$ have multiplicities
$B_{ij}^2$. I call a Coxeter diagram \textit{positive definite} if it
is generated in this way by a positive definite symmetrised Cartan
matrix. The sets of vertices of the connected components of a Dynkin
or Coxeter diagram correspond to the maximal subsets of row indices
connected by sequences $i_0,i_1,\dots,i_n$ as above for any member of
the corresponding isomorphism class of matrices. To display Coxeter
diagrams inline I use a notation where $*$ denotes a vertex and $-$,
$=$ and $\equiv$ lines of multiplicities 1--3. When more than two
lines issue from a vertex, I call this vertex a \textit{node},
following Jacobson's terminology~\cite{ref:Jac62}. Only diagrams with
three single lines issuing from the node will need to be displayed.
Such nodes are denoted by $>\!*\,-$ followed by a subgraph and
preceded by a pair of subgraphs in parentheses. Thus, for example, the
$F_4$ diagram is shown as $*-*=*-*$ and the $E_6$ diagram as
$(*-*,*)>*-*-*$.

\section{\label{sec:an}Analysis}

I now set out to determine which connected Coxeter diagrams are
positive definite. Let it be noticed first that if a symmetrised
Cartan matrix is positive definite then every principal submatrix is
positive definite. Therefore if some subdiagram of a Coxeter diagram
is not positive definite then the entire diagram is not. For my first
two propositions I follow Jacobson~\cite{ref:Jac62}.

\begin{proposition}\label{th:dir}
  For a Coxeter diagram to be positive definite it must have less than
  $l$ lines not counting multiplicity, where $l$ is the order of the
  diagram. In particular it must have no cycles.
\end{proposition}
\begin{proof}
  Let $B$ be a symmetrised Cartan matrix which generates the diagram
  and $v$ the $l$-dimensional row vector with every entry equal to 1.
  Then $v B v^T = \linebreak 2(l - \sum_{i<j} \sqrt{m_{ij}})$ in terms
  of the multiplicities $m_{ij}$ of the lines of the Coxeter diagram.
  Since $m_{ij} \ge 1$ when $m_{ij} > 0$ the sum in this expression
  must have less than $l$ non-zero terms for $B$ to be positive
  definite. For the second part notice that a cycle is a subdiagram
  with as many lines as vertices.
\end{proof}

\noindent It follows from this proposition that arbitrary orientation
of each multiple line of a connected positive definite Coxeter diagram
gives rise to a Dynkin diagram generated by a Cartan matrix which upon
symmetrisation generates the Coxeter diagram. Let indeed the Coxeter
diagram be generated by a symmetrised Cartan matrix $B$ and let its
vertices be indexed by the corresponding row indices of $B$. Let a
line in the diagram be chosen, let $p$ and $q$ be the indices of its
endpoints and let $S_p$ and $S_q$ be the sets of indices of the
vertices of the maximal connected subdiagrams containing the vertices
indexed by $p$ and $q$, respectively, and not containing the chosen
line. Because the Coxeter diagram has no cycles the total set of
indices is the disjoint union of $S_p$ and $S_q$. Consider the matrix
$A$ with entries
\begin{equation}\label{eq:B->A}
  A_{ij} = c_i^{-1} B_{ij} c_j ,
\end{equation}
where $c_i = \sqrt{m_{pq}}$ for $i \in S_p$ and $c_i = 1$ for
$i \in S_q$. Because no line other than the chosen one connects the
subdiagrams with vertices indexed by $S_p$ and $S_q$ this matrix has
entries $A_{pq} = -1$, $A_{qp} = - m_{pq}$ and $A_{ij} = B_{ij}$ when
$\{i,j\} \ne \{p,q\} $. Repeating this for every line in the diagram
results in a Cartan matrix $A$ where for every line in the diagram
$A_{pq} = -1$ and $A_{qp} = - m_{pq}$ in terms of the pair $(p,q)$ of
indices of the endpoints of this line chosen in the construction.
Since the order of $p$ and $q$ in every such pair is arbitrary Cartan
matrices with any orientation of the multiple lines of their Dynkin
diagrams can be obtained in this way. By~\eqref{eq:B->A} these Cartan
matrices are symmetrisable and their symmetrised Cartan matrices equal
$B$. The multiplicities of corresponding lines in the Coxeter and
Dynkin diagrams are identical. If the Coxeter diagram has no multiple
line there is only one corresponding Dynkin diagram, which is
identical to the Coxeter diagram, and one Cartan matrix for each
generating symmetrised Cartan matrix, which is identical to the
latter.

The following observation is used frequently in the following.

\begin{proposition}
  For a Coxeter diagram to be positive definite no vertex must have a
  degree larger than 3 \textnormal{counting} multiplicity.
\end{proposition}
\begin{proof}
  A positive definite real symmetric matrix has positive determinant.
  Let $u$ be a vertex in the Coxeter diagram and consider the
  subdiagram whose vertices are $u$ and its adjacent vertices. Let $l$
  be the order of this subdiagram and let the subdiagram be generated
  by a symmetrised Cartan matrix $B$ whose top row corresponds to the
  vertex $u$. Then, since the Coxeter diagram has no cycles,
  \begin{equation}
    B = \begin{pmatrix}
      2 & -\sqrt{m_{12}} & -\sqrt{m_{13}} & \cdots & -\sqrt{m_{1l}}\\
      -\sqrt{m_{12}} & 2 & 0 & \cdots & 0 \\
      -\sqrt{m_{13}} & 0 & 2 & \cdots & 0 \\
      \vdots & \vdots & \vdots & \ddots & \vdots \\
      -\sqrt{m_{1l}} & 0 & 0 & \cdots & 2      
    \end{pmatrix}
  \end{equation} 
  in terms of the multiplicities $m_{ij}$ of the lines of the
  subdiagram with its vertices indexed by the row indices of $B$. One
  calculates $|B| = 2^{l-2}(4 - \sum_{i=2}^l m_{1i})$. For this to be
  positive one must have $m = \sum_{i=2}^l m_{1i} < 4$, where $m$ is
  the degree of $u$ counting multiplicity.
\end{proof}

\noindent It is an immediate consequence of this proposition that no
other Coxeter diagram with a triple line than the $G_2$ diagram
$*\equiv*$ is positive definite.

My main tool is the following.

\begin{lemma}
  Given a tridiagonal, symmetric, real, $l \times l$ matrix $T$ with
  every diagonal entry equal to 2 let the sequence $p_i,i=0 \dots l,$
  be defined by
  \begin{equation}\begin{gathered}
    p_0 = 1, \quad p_1= 2 , \quad
    p_i = 2 p_{i-1} - T_{i,i-1}^2 p_{i-2}, \quad i \ge 2 .
  \end{gathered}\end{equation}
  Then $T$ is positive definite if and only if every $p_i$ is
  positive.
\end{lemma}
\begin{proof}
  An easy induction by means of the cofactor expansion of determinants
  shows that $p_i$ is for $i > 0$ the leading principal minor of $T$
  of order $i$. The lemma thus expresses Sylvester's criterion,
  e.g.~\cite{ref:Hor85}, for positive definiteness of a real symmetric
  matrix.
\end{proof}
\begin{remark}
  If $T$ is a symmetrised Cartan matrix then
  $T_{i,i-1}^2 = m_{i,i-1}$, so the sequence of these coefficients can
  then be read directly off its Coxeter diagram when the vertices of
  the latter are ordered by the row indices of $T$.
\end{remark}

\noindent A linear Coxeter diagram is generated by some tridiagonal
symmetrised Cartan matrix. The lemma therefore allows me to deal at
one stroke with a large class of such diagrams.

\begin{proposition}\label{th:lin}
  All $A_l$ Coxeter diagrams $*-*-*\cdots$ with $l \ge 1$ are positive
  definite. So are the $B_l/C_l$ Coxeter diagrams $*=*-*-*\cdots$ with
  $l \ge 2$. The $F_4$ diagram $*-*=*-*$ is the only Coxeter diagram
  with chains of single lines issuing from both endpoints of a double
  line that is positive definite. The $G_2$ Coxeter diagram $*\equiv*$
  is positive definite
\end{proposition}
\begin{proof}
  The sequences of minors $p_i, i \ge 2,$ of the lemma are calculated
  easily from the Coxeter diagrams when their vertices are ordered
  from the left, cf. the remark. One gets
  \begin{equation}\label{eq:lin}\begin{gathered}
    A_l: 3,4,5,\ldots, \quad B_l/C_l: 2,2,2,\ldots, \\
    *-*=*-*-*\cdots: 3,2,1,0,-1,\ldots \quad G_2: 1.
  \end{gathered}\end{equation}
It is straightforward to verify by formal induction the systematics
visualised in~\eqref{eq:lin}. In the cases of $A_l$ and $B_l/C_l$ the
order $l$ of a positive definite Coxeter diagram is seen to have no
upper limit. In the case of diagrams $*-*=*-*-*\cdots$ the minor $p_i$
is seen to cease to be positive for $i = 5$. It follows that $l = 4$
is the maximal order of this type of Coxeter diagram if it is to be
positive definite. This also rules out diagrams with chains of more
that one single line issuing from both endpoint of a double line
because such diagrams would have a subdiagram $*-*=*-*-*$. The $G_2$
Coxeter diagram is seen to be positive definite.
\end{proof}

\begin{proposition}\label{th:max1}
  For a connected Coxeter diagram to be positive definite it
  must have at most one double line.
\end{proposition}

\begin{proof}
  If the diagram had two double lines it would have a subdiagram
  \linebreak $*=*-*\cdots*-*=*$ with one or more single lines between
  the double lines. Calculation for this subdiagram of the minors
  $p_i$ of the lemma with the vertices ordered from the left gives
  $p_i = 2,2,\dots,2,0$ for $2 \le i \le l$, where $l$ is the order of
  the diagram. Its Coxeter diagram therefore is not positive definite.
\end{proof}

I next turn to Coxeter diagrams with a node and first consider
diagrams $(*,*)>*-*\cdots$ where $*\cdots$ denotes some subdiagram
generated by a symmetrised Cartan matrix $Z$ whose top row corresponds
to the leftmost vertex. A symmetrised Cartan matrix which generates
the total diagram is
\begin{equation}
  B =\begin{pmatrix}
    X & Y \\
    Y^T  & Z
  \end{pmatrix} ,
\end{equation}
where
\begin{equation}
  X = \begin{pmatrix}
    2 & 0 & -1 \\
    0 & 2 & -1 \\
    -1 & -1 & 2
  \end{pmatrix}
\end{equation}
and the entries of $Y$ are 0 except for $-1$ in its lower-left corner.
The transformation $B \mapsto B' = T B T^T$ with
\begin{equation}
  T = \begin{pmatrix}
    S & 0 \\
    0  & I
  \end{pmatrix} ,
\end{equation}
where
\begin{equation}
  S = \begin{pmatrix}
    \sqrt{\tfrac12} & -\sqrt{\tfrac12} & 0 \\
    \sqrt{\tfrac12} & \sqrt{\tfrac12} & 0 \\
    0 & 0 & 1 \\
  \end{pmatrix}
\end{equation}
and $I$ is the identity matrix of the same dimension as $Z$, gives
\begin{equation}
  B' =\begin{pmatrix}
    X' & Y \\
    Y^T  & Z
  \end{pmatrix}
\end{equation}
with
\begin{equation}
  X' =\begin{pmatrix}
    2 & 0 & 0 \\
    0 & 2 & -\sqrt{2} \\
    0 & -\sqrt{2}  & 2
  \end{pmatrix} .
\end{equation}
The submatrix $Y$ does not change in this transformation because $S$
has zeros in its last row except for 1 on the diagonal. The matrix
$B'$ is recognised as a symmetrised Cartan matrix generating the
disconnected Coxeter diagram \linebreak $*\ *=*-*\cdots$. Hence the
diagram $(*,*)>*-*\cdots$ is positive definite if and only if the
diagram $*=*-*\cdots$ is positive definite. Two consequences can be
drawn immediately.

\begin{proposition}
  The $D_l$ Coxeter diagrams $(*,*)>*-*-*\cdots$ of order $l \ge 4$
  are positive definite.
\end{proposition}
\begin{proof}
  This follows from the $B_l/C_l$ diagrams $*=*-*-*\cdots$ of order
  $l \ge 3$ being positive definite.
\end{proof}

\begin{proposition}\label{th:max1'}
  For a connected Coxeter diagram to be positive definite it must have
  at most one double line or node.
\end{proposition}
\begin{proof}
  This follows from proposition~\ref{th:max1} and the observation
  immediately before the preceding proposition.
\end{proof}

A similar reasoning can be applied to Coxeter diagrams
$(*-*,*)>*-*\cdots$. By proposition~\ref{th:max1'} it suffices to
consider the case when the subdiagram $*\cdots$ is an $A_l$ diagram.
With $B$, $B'$ and $T$ expressed in block form as before, now
\begin{equation}
  X = \begin{pmatrix}
    2 & -1 & 0 & 0 \\
    -1 & 2 & 0 & -1 \\
    0 & 0 & 2 & -1 \\
    0 & -1 & -1 & 2
  \end{pmatrix}
\end{equation}
and
\begin{equation}
  S = \begin{pmatrix}
    0 & \sqrt{\tfrac12} & -\sqrt{\tfrac12} & 0 \\
    1 & 0 & 0 & 0 \\
    0 & \sqrt{\tfrac12} & \sqrt{\tfrac12} & 0 \\
    0 & 0 & 0 & 1
  \end{pmatrix} .
\end{equation}
This gives
\begin{equation}
  X' = \begin{pmatrix}
    2 & -\sqrt{\tfrac12} & 0 & 0 \\
    -\sqrt{\tfrac12} & 2 & -\sqrt{\tfrac12} & 0 \\
    0 & -\sqrt{\tfrac12} & 2 & -\sqrt{2}  \\
    0 & 0 & -\sqrt{2} & 2
  \end{pmatrix} .
\end{equation}
The matrix $B'$ with this submatrix $X'$ is recognised as a
symmetrised Cartan matrix generating a generalised Coxeter diagram
$*\sim*\sim*=*-*\cdots$ where lines $\sim$ have ``multiplicities''
$\frac12$. Calculation of the minors $p_i$ of the lemma gives
$p_i = \tfrac72,6,5,4,3,2,1,0,-1,\cdots$ for $i \ge 2$ when the
vertices are ordered from the left. Here $p_i$ ceases to be positive
for $i = 9$, so:

\begin{proposition}
  The $E_l$ diagrams $(*-*,*)>*-*-*\cdots$ with $l \ge 6$ are positive
  definite if and only if $l \le 8$.
\end{proposition}

Consider finally Coxeter diagrams $(*-*,*-*)>*-*\cdots$ where again
the subdiagram $*\cdots$ is an $A_l$ diagram. Now
\begin{equation}
  X = \begin{pmatrix}
    2 & 0 & -1 & 0 & 0 \\
    0 & 2 & 0 & -1 & 0 \\
    -1 & 0 & 2 & 0 & -1 \\
    0 & -1 & 0 & 2 & -1 \\
    0 & 0 & -1 & -1 & 2
  \end{pmatrix}
\end{equation}
and
\begin{equation}
  S = \begin{pmatrix}
    0 & 0 & \sqrt{\tfrac12} & -\sqrt{\tfrac12} & 0 \\
    \sqrt{\tfrac12} & -\sqrt{\tfrac12} & 0 & 0 & 0 \\
    \sqrt{\tfrac12} & \sqrt{\tfrac12} & 0 & 0 & 0 \\
    0 & 0 & \sqrt{\tfrac12} & \sqrt{\tfrac12} & 0 \\
    0 & 0 & 0 & 0 & 1
  \end{pmatrix} .
\end{equation}
This gives
\begin{equation}
 X' = \begin{pmatrix}
    2 & -1 & 0 & 0 & 0 \\
    -1 & 2 & 0 & 0 & 0 \\
    0 & 0 & 2 & -1 & 0 \\
    0 & 0 & -1 & 2 & -\sqrt{2} \\
    0 & 0 & 0 & -\sqrt{2} & 2
  \end{pmatrix} .
\end{equation}
The matrix $B'$ with this submatrix $X'$ is recognised as a
symmetrised Cartan matrix generating the disconnected Coxeter diagram
$*-*\ *-*=*-*\cdots$, which is positive definite if and only if the
diagram $*-*=*-*\cdots$ is so. But by proposition~\ref{th:lin} this is
not the case when the order of $*\cdots$ is greater than 1. Hence:

\begin{proposition}
  Coxeter diagrams where three chains of single lines, each of length
  at least 2, issue from a node are not positive definite.
\end{proposition}

\section{\label{con}Conclusion}

The propositions in section~\ref{sec:an} can be summarised in the
following:

\begin{theorem}
  The positive definite connected Coxeter diagrams are: $A_l$ for
  $l \ge 1$, $B_l/C_l$ for $l \ge 2$, $D_l$ for $l \ge 4$, $E_l$ for
  $6 \le l \le 8$, $F_4$ and $G_2$.
\end{theorem}

\noindent By the observation after proposition~\ref{th:dir} arbitrary
orientation of each multiple line in a diagram in the list gives a
Dynkin diagram of a reduced root system. Due to the symmetry of the
Coxeter diagram in the remaining cases distinct Dynkin diagrams only
arise in this way in the case of the $B_l/C_l$ diagrams with
$l \ge 3$. Each of these Coxeter diagrams corresponds to two Dynkin
diagrams, $B_l$ with Cartan matrix entry $A_{12} = -2$ and $C_l$ with
Cartan matrix entry $A_{21} = -2$ when the Cartan matrices are chosen
such that their first and second rows correspond to the first and
second vertices from the left in the Coxeter diagram $*=*-*\cdots$. By
the observations in sections~\ref{sec:in}--\ref{sec:sym} one arrives
at the familiar result~\cite{ref:Kil88,ref:Car94,ref:Dyn47}:

\begin{theorem}[Killing, Cartan, Dynkin]
  The Dynkin diagrams of simple finite-dimensional Lie algebras over
  $\mathbb C$ are: $A_l$ for $l \ge 1$, $B_l$ for $l \ge 2$, $C_l$ for
  $l \ge 3$, $D_l$ for $l \ge 4$, $E_l$ for $6 \le l \le 8$, $F_4$ and
  $G_2$.
\end{theorem}

\end{document}